%% file: hybrid_autoregressive_transducer.tex
\documentclass{article}
\usepackage{spconf,amsmath,graphicx}
\usepackage{mathtools}
\usepackage{amssymb,mathrsfs}
\usepackage{amsfonts,amsthm}
\usepackage{enumitem}
\usepackage{subcaption}
\usepackage{bbm}
\usepackage{xcolor}

\usepackage[T1]{fontenc}
\usepackage{etoolbox}

\DeclareMathOperator*{\argmax}{argmax}
\newcommand{\ignore}[1]{}
\newcommand{\pp}[3]{\phi_{#2,#3}(#1)}
\newtheorem{proposition}{Proposition}
\newtheorem{lemma}{Lemma}
\newtheorem{corollary}{Corollary}

\newcommand{\fullpaper}[1]{#1}
\newcommand{\shortpaper}[1]{}

\patchcmd{\abstract}{\small}{}{}{}
\patchcmd{\conclusion}{\small}{}{}{}

\title{Hybrid Autoregressive Transducer (HAT)}
\name{Ehsan Variani, David Rybach, Cyril Allauzen, Michael Riley
\thanks{equal contribution}}
\address{\texttt{$\{$variani, rybach, allauzen, riley$\}$@google.com}
}
\begin{document}
%
\maketitle
\begin{abstract}
\input{abstract}

\end{abstract}
\input{introduction}
\input{time_synchronous_models}
\input{hat_model}

\input{experiment}
\input{conclusion}

\vfill\pagebreak

\label{sec:refs}
\bibliographystyle{IEEEbib}
\bibliography{strings,refs}

\vfill\pagebreak

\appendix
\input{appendix}
\end{document}

%% file: abstract.tex
This paper proposes and evaluates the {\em hybrid autoregressive transducer
(HAT) model}, a time-synchronous encoder-decoder model that preserves the modularity
of conventional automatic speech recognition systems.
The HAT model provides a way to measure the quality of the internal language model
that can be used to decide whether inference with an
external language model is beneficial or not.
\fullpaper{
This article also presents a
finite context version of the HAT model that addresses the exposure bias problem and
significantly simplifies the overall training and inference.
}
We evaluate our proposed model on a large-scale voice search task. Our
experiments show significant improvements in WER compared to the state-of-the-art
approaches
\shortpaper{
\footnote{Please note that this is a summary of our paper, for details
please visit the arxiv submission.}
}
.

%% file: introduction.tex
\shortpaper{
\vspace{-0.3cm}
}
\section{Introduction}
\label{sec:intro}
\shortpaper{
\vspace{-0.3cm}
}
The automatic speech recognition (ASR) problem is probabilistically formulated as a
\textit{maximum a posteriori} decoding
\cite{jelinek1976continuous,bahl1983maximum}:
\begin{eqnarray}
\hat{W} = \argmax_{W} P( W | X) = \argmax_{W} P( X | W) \cdot  P(W)
\end{eqnarray}
where $X$ is a sequence of acoustic features and $W$ is a candidate
word sequence. In conventional modeling, the
posterior probability is factorized as a product of an acoustic likelihood
and a prior \textit{language model} (LM)
via Bayes' rule. The acoustic likelihood is further factorized as a product
of conditional distributions of acoustic features given a phonetic sequence,
the \textit{acoustic model} (AM), and conditional distributions of a phonetic sequence
given a word sequence, the \textit{pronunciation model}.
The phonetic model structure, pronunciations and LM are commonly represented and
combined as {\em weighted finite-state transducers (WFSTs)} 
\cite{mohri2002weighted,mohri2008}.

The probabilistic factorization of ASR allows a \textit{modular} design
that provides practical advantages for training and inference.
The modularity permits training acoustic and language models independently and
on different data sets. A decent acoustic model can be trained
with a relatively small amount of transcribed audio data whereas the LM can be
trained on text-only data, which is available in vast amounts for many languages.
Another advantage of the modular modeling approach is it allows dynamic modification
of the component models in well proven, principled methods such as
vocabulary augmentation \cite{allauzen2015}, LM adaptation and contextual biasing \cite{Hall2015}.

From a modeling perspective, the modular factorization of the posterior probability
has the drawback that the
parameters are trained separately with different objective functions. In
conventional ASR this has been addressed to some extent by introducing
discriminative approaches for acoustic modeling
\cite{bahl1986maximum, povey2005discriminative, kingsbury2009lattice} and
language modeling
\shortpaper{
\cite{chen2000discriminative}.
}
\fullpaper{
\cite{chen2000discriminative, kuo2002discriminative}.
}
These optimize the posterior probability by varying the parameters of
one model while the other is frozen.
\ignore{
These algorithms maximize the posterior probability either explicitly as in
Maximum Mutual Information (MMI) \cite{bahl1986maximum},
or by introducing a surrogate loss function as in 
Minimum Phone Error (MPE) \cite{povey2005discriminative} or state-level
Minimum Bayes Risk (sMBR) \cite{kingsbury2009lattice}. In these approaches one
of the factors (LM in discriminative acoustic modeling and AM in discriminative
language modeling) is kept frozen while the other factor parameters are
trained to maximize the posterior probability on the set of transcribed audio data.
}

Recent developments in the neural network literature, particularly
{\em sequence-to-sequence (Seq2Seq) modeling}
\shortpaper{
\cite{auli2013joint, kalchbrenner2013recurrent},
}
\fullpaper{
\cite{auli2013joint, kalchbrenner2013recurrent,cho2014learning,cho2014properties,sutskever2014sequence},
}
have allowed the full optimization of the posterior probability, 
learning a direct mapping between feature sequences and
orthographic-based (so-called {\em character-based}) labels, like graphemes or wordpieces.
One class of these models uses a {\em recurrent neural network (RNN)} architecture such as the
{\em long short-term memory (LSTM)} model \cite{hochreiter1997long} 
followed by a softmax
layer to produce label posteriors. All the parameters are optimized at the sequence
level using, e.g., 
the {\em connectionist temporal classification (CTC)} \cite{graves2006connectionist} criterion.
The {\em letter models}
\shortpaper{
\cite{eyben2009speech}
}
\fullpaper{
\cite{eyben2009speech,hannun2014deep, collobert2016wav2letter}
}
and {\em word model}
\cite{soltau2016neural} are examples of this neural model class.
Except for the modeling unit, these models are very similar to conventional
acoustic models and perform well when combined
with an external LM during decoding (beam search).
\shortpaper{
\cite{zenkel2017comparison}.
}
\fullpaper{
\cite{zenkel2017comparison,battenberg2017exploring}.
}

Another category of direct models are the {\em encoder-decoder} networks. Instead of
passing the encoder activations to a softmax layer for the label posteriors,
these models use a decoder network to combine the encoder information with 
an embedding of the previously-decoded label history that provides
the next label posterior.
The encoder and decoder parameters are optimized together using
a time-synchronous objective function as with the {\em RNN transducer} (RNNT)
\cite{graves2012sequence} or a label-synchronous objective as with {\em listen, attend,
and spell (LAS)} \cite{chan2016listen}. These models are usually trained with
character-based units and decoded with a basic beam search. There has been
extensive efforts to develop decoding algorithms that can use external
LMs, so-called {\em fusion methods}
\shortpaper{
\cite{gulcehre2015using,chorowski2016towards}.
}
\fullpaper{
\cite{gulcehre2015using,chorowski2016towards,sriram2017cold,hori2018end,shan2019component}.
}
However, these methods have shown relatively small gains on large-scale ASR tasks
\cite{kannan2018analysis}. In the current state of encoder-decoder ASR models,
the general belief is that encoder-decoder models with character-based units outperform
the corresponding phonetic-based models. This has led some to conclude
a lexicon or external LM is unnecessary
\cite{sainath2018no,zhou2018comparison,irie2019model}.

This paper proposes the {\em hybrid autoregressive transducer (HAT)}, a time-synchronous
encoder-decoder model which couples the powerful probabilistic capability of
Seq2Seq models with an inference algorithm that preserves
modularity and external lexicon and LM integration. The HAT model allows an
internal LM quality measure, useful to decide if an external language model
is beneficial. The finite-history version of the HAT model is also presented
and used to show how much label-context history is needed to train a
state-of-the-art ASR model on a large-scale training corpus.

%% file: time_synchronous_models.tex
\shortpaper{
\vspace{-0.3cm}
}
\section{Time-Synchronous Estimation of Posterior: Previous Work}
\label{sec:time_sync_review}
\shortpaper{
\vspace{-0.2cm}
}

\fullpaper{
\begin{figure}[t]
  \centering
  \hspace{-0.7cm}
  \includegraphics[width=1.0\linewidth,keepaspectratio]{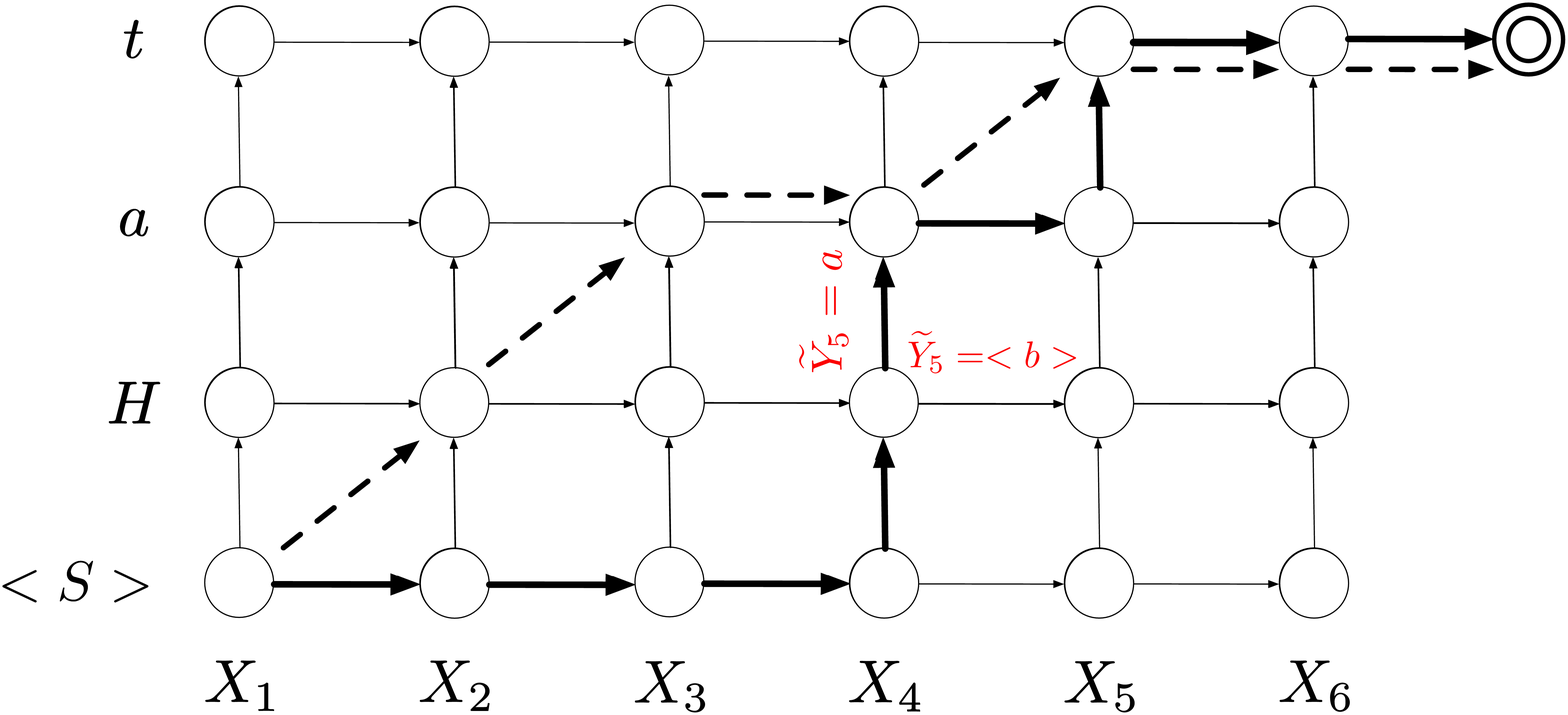}
  \caption{Lattice encoding all the eligible alignment paths. Each alignment
  path starts from bottom-left corner and ends in top-right corner. The path
  can traverse horizontal, vertical or diagonal edges.
  }
  \label{fig:lattice}
\end{figure}
}
\shortpaper{
\begin{figure}[t]
  \centering
  \hspace{-0.5cm}
  \includegraphics[width=.7\linewidth,keepaspectratio]{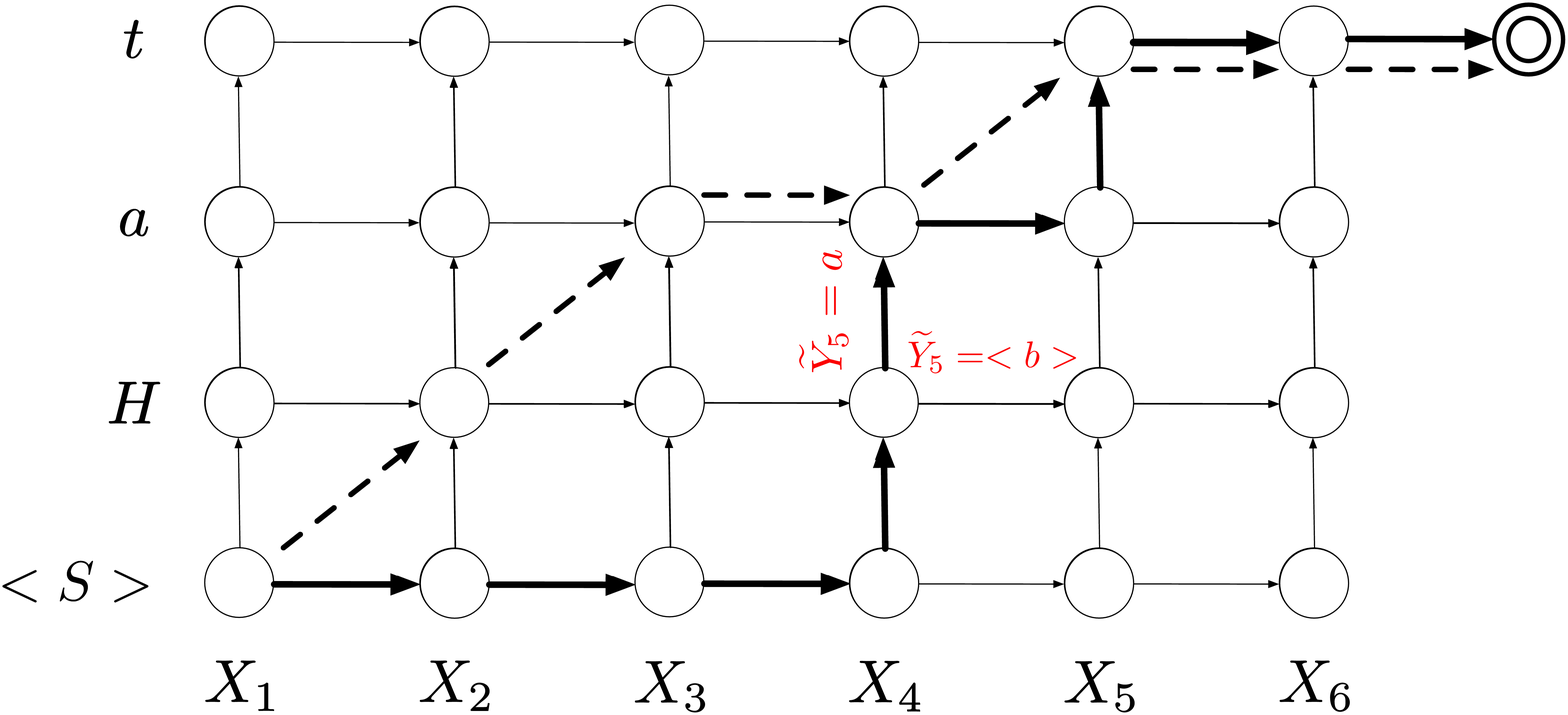}
  \caption{Lattice encoding all the eligible alignment paths.}
  \label{fig:lattice}
  \vspace{-0.4cm}
\end{figure}
}
For an acoustic feature sequence $X = X_{1:T}$ corresponding to a word sequence
$W$, assume $Y=Y_{1:U}$ is a tokenization of $W$ where $Y_i \in V$, is
either a phonetic unit or a character-based unit from a finite-size alphabet $V$.
The character tokenization of transcript
\texttt{<S> Hat} and the corresponding acoustic feature sequence $X_{1:6}$
are depicted 
in the vertical and the horizontal axis of Figure~\ref{fig:lattice},
respectively. Define an {\em alignment path} $\widetilde{Y}$ as a sequence
of \textit{edges} ${\widetilde{Y}}_{1:|\widetilde{Y}|}$ where $|\widetilde{Y}|$,
the \textit{path length}, is the
number of edges in the path. Each path passes through sequence of nodes
$\left(X_{t_i}, Y_{u_i} \right)$ where $0 \leq t_i \leq T$, $0 \leq u_i \leq U$
and $t_i \geq t_{i-1}$, $u_i \geq u_{i-1}$.
The dotted path and the
bold path in Figure~\ref{fig:lattice} are two left-to-right alignment
paths between $X_{1:6}$ and $Y_{0:3}$. For any sequence pairs $X$ and $Y$,
there are an exponential number of alignment paths. Each time-synchronous model
typically defines 
a subset of these paths as the permitted path inventory, which
is usually referred to as the path \textit{lattice}. Denote
by $B: \widetilde{Y} \rightarrow {Y}$ the function that maps
permitted alignment paths to the corresponding label sequence.
For most models, 
this is many-to-one.
Besides the definition of the lattice, 
time-synchronous models usually differ on how
they define the alignment path posterior probability $P ( \widetilde{Y} | X )$.
Once defined, the posterior probability for a time-synchronous model
is calculated by
marginalizing over these alignment path posteriors:
\shortpaper{
\vspace{-0.3cm}
\begin{eqnarray}
  P \left( Y | X \right) = \sum_{\widetilde{Y}: {B}( \widetilde{Y} ) = Y} P ( \widetilde{Y} | X)
  \label{eq:margin_posterior}
  \vspace{-0.65cm}
\end{eqnarray}
  \vspace{-0.1cm}
}
\fullpaper{
\begin{eqnarray}
  P \left( Y | X \right) = \sum_{\widetilde{Y}: {B}( \widetilde{Y} ) = Y} P ( \widetilde{Y} | X)
  \label{eq:margin_posterior}
\end{eqnarray}
}
The model parameters are optimized by maximizing $P \left( Y | X \right)$.
\vspace{0.04cm}

\shortpaper{
\noindent {\textbf{Cross-Entropy (CE)}} and
          {\textbf{Connectionist Temporal Classification (CTC):}} These models make
          conditional independence assumption to calculate alignment posteriors. The CE models
          lattice contains only one path, while CTC lattice is defined by the
          topology explained in \cite{graves2006connectionist}.
}
\fullpaper{
\noindent {\textbf{Cross-Entropy (CE):}}
In the cross-entropy models, the alignment lattice contains only one path
which is
derived by a pre-processing step known as forced alignment. These models encode
the feature sequence $X$ via a stack of RNN layers to output activations
${\bf f \left( X \right)} = {\bf f_{1:T}}$. The activations are then passed to a linear
layer (aka logits layer) to
produce the unnormalized class-level scores, $S ( {\widetilde{Y}}_t = {\tilde{y}}_t | X )$.
The conditional posterior $P ( {\widetilde{Y}}_t| X )$ is derived by normalizing
the class-level scores using the softmax function
\shortpaper{.
}
\fullpaper{
\[
P ( {\widetilde{Y}}_t = y | X ) = \frac{\exp ( S( {\widetilde{Y}}_t = {\tilde{y}}_t | X ))} {\sum_{{{\tilde{y}}'}_t } \exp(S( {\widetilde{Y}}_t = {{\tilde{y}}'}_t | X ))}
\]
}
The alignment path posterior is derived by imposing the conditional independence assumption between label
sequence $Y$ given feature sequence $X$:
\shortpaper{$P\left( Y | X \right) = P ( \widetilde{Y} | X ) = \prod_{t=1}^{T} P ( {\widetilde{Y}}_t = {\tilde{y}}_t | X)$.
}
\fullpaper{
\[
P\left( Y | X \right) = P ( \widetilde{Y} | X ) = \prod_{t=1}^{T} P ( {\widetilde{Y}}_t = {\tilde{y}}_t | X)
\]
}
For inference with an external LM, the search is defined as:
\shortpaper{\vspace{-0.2cm}
}
\begin{eqnarray}
  {\tilde{y}}^{\ast} = \argmax_{\tilde{y}} \,\, \lambda_1 \phi ( x | \tilde{y} ) +  \log P_{LM} (B(\tilde{y}))
\end{eqnarray}
where $\lambda_1$ is a scalar and
$\phi ( x | \tilde{y} ) = \log (\prod_{t=1}^T P ( x_t | {\tilde{y}}_t ))$ is a
pseudo-likelihood sequence-level score derived by applying Bayes' rule on the time
posteriors: $P ( x_t | {\tilde{y}}_t ) \propto P ({\tilde{y}}_t | x_t) / P({\tilde{y}}_t)$
where $P({\tilde{y}}_t)$ is a label prior \cite{morgan1990continuous}.
\ignore{
\[
  P \left( X | Y \right) = \frac {\prod_{t=1}^T P \left( {\tilde{Y}}_t | X \right)}
  {\prod_{t=1}^T P \left( {\tilde{Y}}_t = y_t \right)} = \prod_{t=1}^T \frac{P \left( {\tilde{Y}}_t | X \right)}{P \left( {\tilde{Y}}_t = y_t \right)}
\]
}

\shortpaper{
\vspace{0.1cm}
}
\fullpaper{
\vspace{0.2cm}
}
\noindent {\textbf{Connectionist Temporal Classification (CTC):}}
The CTC model \cite{graves2006connectionist} augments the label alphabet with a
blank symbol, \texttt{<b>}, and defines
the alignment lattice to be the grid of all paths with edit distance of $0$ to the
ground truth label sequence after removal of all blank symbols and consecutive
repetitions.
\ignore{
For example, if the sequence of symbols from ground truth alignment
is \texttt{|h a a t t|}, the resulting
CTC path inventory should encodes all the paths
with the following pattern: \texttt{|<b>$\ast$ h$+$ <b>$\ast$ a$+$ <b>$\ast$ t$+$ <b>$\ast$|} where $\ast$ means
$0$ or more repetitions, and $+$ means $1$ or more repetitions.
}
The dotted path in
Figure~\ref{fig:lattice} is a valid CTC path.
Similar to the CE models, the CTC models also use a stack of RNNs followed by a logits layer
and softmax layer to produce local posteriors. The sequence-level alignment posterior is then
calculated by making a similar conditional independence assumption between 
the alignment label
sequence given $X$,
$
P ( \widetilde{Y} | X )=\prod_{t=1}^{T} P ({\widetilde{Y}}_t = {\tilde{y}}_t | X )
$.
Finally $P( Y | X )$ is calculated by marginalizing over 
the alignment posteriors with Eq~\ref{eq:margin_posterior}.
}
\fullpaper{
\vspace{0.2cm}
}

\noindent {\textbf{Recurrent Neural Network Transducer (RNNT):}}
The RNNT lattice encodes
all the paths of length $T+U-1$ edges starting from the bottom-left corner of the lattice in
Figure~\ref{fig:lattice} and ending in the top-right corner.
Unlike CE/CTC
lattice, RNNT allows staying in the same time frame and emitting multiple labels, like the
two consecutive vertical edges in the bold path in Figure~\ref{fig:lattice}.
Each alignment
path $\widetilde{Y}$ can be denoted by a sequence of labels 
${\widetilde{Y}}_{k}$, $1 \leq k \leq T + U -1$ prepended by ${\widetilde{Y}}_{0} = \texttt{<S>}$.
Here ${\widetilde{Y}}_{k}$ is a random
variable representing the edge label corresponding to time frame
$t=h({\widetilde{Y}}_{1:k-1})\triangleq \sum_{i=1}^{k-1} {\mathbbm{1}}_{{\widetilde{Y}}_{i} = \texttt{<b>}}$
while the previous
$u=v({\widetilde{Y}}_{1:k-1})\triangleq \sum_{i=1}^{k-1} {\mathbbm{1}}_{{\widetilde{Y}}_{i} \neq \texttt{<b>}}$ 
labels have been already
emitted. The indicator function ${\mathbbm{1}}_{{\widetilde{Y}}_{i} = \texttt{<b>}}$ returns
$1$ if ${\widetilde{Y}}_{i} = \texttt{<b>}$ and $0$ otherwise.
For the bold path in Figure~\ref{fig:lattice},
$\widetilde{Y} = \texttt{<S>}, \texttt{<b>}, \texttt{<b>}, \texttt{<b>}, \texttt{H}, \texttt{a}, \texttt{<b>}, \texttt{t}, \texttt{<b>}$.
\ignore{
$\widetilde{Y} = {\tilde{Y}}_{1, 0}, {\tilde{Y}}_{2, 0}, {\tilde{Y}}_{3, 0},
{\tilde{Y}}_{4, 1}, {\tilde{Y}}_{4, 2}, {\tilde{Y}}_{5, 2}, 
{\tilde{Y}}_{5, 3}, {\tilde{Y}}_{5, 4}, {\tilde{Y}}_{6, 4}$.
At each position $\left(t,u \right)$, the path prefix is denoted by
$\pp{\tilde{Y}}{t}{u}$.
For the bold path of Figure~\ref{fig:lattice},
$\pp{\tilde{Y}}{4}{1} = 
{\tilde{Y}}_{1, 0}, {\tilde{Y}}_{2, 0}, {\tilde{Y}}_{3, 0}$.
}

The local posteriors in RNNT are calculated under the assumption that
their value is independent of the path prefix, i.e.,
if
$
{B}({\widetilde{Y}}_{1:k-1} ) = {B}({\widetilde{Y}'}_{1:k-1} )
$
, then
$
P ( {\widetilde{Y}}_{k} | X, {\widetilde{Y}}_{1:k-1})
=
P ( {\widetilde{Y}}'_{k} | X, {\widetilde{Y'}}_{1:k-1})
$.
The local posterior is calculated using the encoder-decoder architecture.
The encoder accepts input $X$ and outputs $T$ vectors ${\bf f\left(X\right)} = {\bf f_{1:T}}$. The decoder
takes the label sequence $Y$ and outputs $U + 1$ vectors ${\bf g \left( Y \right)} = {\bf g_{0:U}}$ ,
where the $0^{th}$ vector corresponds to the \texttt{<S>} symbol.
The unnormalized score of the next label ${\widetilde{Y}}_{k}$ corresponding to position
$\left(t,u \right)$ is defined as:
$S  ( {\widetilde{Y}}_{k} = {\tilde{y}}_{k} | X, {\widetilde{Y}}_{1:k-1} ) = {\bf J} \left( {\bf f_t + g_u} \right)$.
The \textit{joint network} ${\bf J}$ is usually a multi-layer non-recurrent network.
The local posterior
$P ( {\widetilde{Y}}_{k} = {\tilde{y}}_{k} | X, {\widetilde{Y}}_{1:k-1} )$
is calculated
by normalizing the above score over all labels ${\tilde{y}}_{k} \in V \cup \texttt{<b>}$ using
the softmax function.
\ignore{
\[
  P ( {\tilde{Y}}_{t, u} = {\tilde{y}}_{t, u} | X, \pp{\tilde{Y}}{t}{u} ) = 
  \frac{\exp \left(Score (t, u, y_{t,u} ) \right)} {\sum_{{{\tilde{y}}'}_{t,u}} \exp \left( \ Score (t, u, {{\tilde{y}}'}_{t,u} ) \right)} \nonumber
  \label{eq:label_posterior}
\]
}
The alignment path posterior is derived by chaining the above quantity over the path,
\shortpaper{
$P({\widetilde{Y}} | X) = \prod_{k=1}^{T+U-1} P ( {\widetilde{Y}}_{k} | X, {\widetilde{Y}}_{1:k-1} )$.
}
\fullpaper{
\[
P({\widetilde{Y}} | X) = \prod_{k=1}^{T+U-1} P ( {\widetilde{Y}}_{k} | X, {\widetilde{Y}}_{1:k-1} )
\]
}
\ignore{
\[
  P ( {\tilde{y}}_{1, 0} | X ) P ( {\tilde{y}}_{2, 0} | X, {\tilde{y}}_{1, 0} ) \cdots
  P ( {\tilde{y}}_{6, 3} |X,  {\tilde{y}}_{1, 0}, \cdots {\tilde{y}}_{5, 3} )
  \]
}
\ignore{
The total label posterior is calculated by
marginalized over all alignment path posteriors using Eq~\ref{eq:margin_posterior}.
}

For the CTC and RNNT models, the inference with an 
external LM is usually formulated as the following search problem:
\shortpaper{
\vspace{-0.25cm}
}
\begin{eqnarray}
  {\tilde{y}}^{\ast} = \argmax_{\tilde{y}} \,\, \lambda_1 \log P' ( \tilde{y} | x ) + \log P_{LM}(B(\tilde{y})) + \lambda_2 v({\tilde{y}})
  \nonumber
\end{eqnarray}
where $P'$ is derived by scaling the blank posterior in $P$ followed
by re-normalization \cite{sak2015fast},
$\lambda_1$ is a scalar weight,
and $\lambda_2$ is a coverage penalty scalar which weights non-blank labels $v({\tilde{y}})$
\cite{chorowski2016towards}.

\shortpaper{
\vspace{-0.4cm}
}

%% file: hat_model.tex
\section{Hybrid Autoregressive Transducer}
\label{sec:hat_model}
\shortpaper{
\vspace{-0.2cm}
}
The Hybrid Autoregressive Transducer (HAT) model is a time-synchronous encoder-decoder model
which
distinguishes itself from other time-synchronous models by 
1) formulating the local posterior probability differently,
2) providing a measure of its internal language model quality,
and 3) offering a
 mathematically-justified inference algorithm for external LM
integration.
\shortpaper{
\vspace{-0.4cm}
}
\subsection{Formulation}
\shortpaper{
\vspace{-0.2cm}
}
The RNNT lattice definition is used for the HAT 
model
\shortpaper{.\footnote{mainly for ease of experimental comparison}}
\fullpaper{
Alternative lattice functions can be explored, this choice is merely selected to
provide fair comparison of HAT
and RNNT models in the experimental section.
}
To calculate the local conditional posterior, the HAT model differentiates between
horizontal and vertical edges in the lattice.
For edge $k$ corresponding to position $\left(t, u\right)$,
the posterior probability to take a horizontal edge
and emitting \texttt{<b>} is modeled by a Bernoulli distribution $b_{t,u}$ which is a
function of the entire past
history of labels and features. The vertical move is modeled by a posterior distribution
$P_{t,u} ( y_{u + 1} | X, y_{0:u})$, defined over labels in $V$.
The alignment posterior
$P( {\widetilde{Y}}_{k} = {\tilde{y}}_{k} | X, 
{\widetilde{Y}}_{1:k-1})$
is formulated as:
\begin{eqnarray}
  \begin{cases*}
    b_{t,u} & ${\tilde{y}}_{k} = <b>$ \\
      \left(1 - b_{t,u} \right) P_{t,u}\left(  {{y}}_{u + 1}  | X, y_{0:u} \right)        & ${\tilde{y}}_{k} = {{y}}_{u + 1}$
  \end{cases*}
  \label{eq:hat_model}
\end{eqnarray}

\noindent{\textbf{Blank (Duration) distribution:}} The input
feature sequence $X$ is fed to a stack of RNN layers to output $T$ encoded vectors
${\bf {f^1}\left( X \right) = {f}^1_{1:T}}$. The label sequence is also fed to
a stack of RNN layers to output ${\bf {g}^1\left( Y \right) = {g}^1_{1:U}}$. The
conditional Bernoulli distribution $b_{t, u}$ is then calculated as:
\shortpaper{
\vspace{-0.3cm}
}
\begin{eqnarray}
  b_{t, u} = \sigma \left({\bf w} \bullet ({\bf {f}^1_t + {g}^1_u}) + {\bf b} \right)
\end{eqnarray}
where $\sigma(\cdot)$ is the sigmoid function, ${\bf w}$ is a weight vector,
$\bullet$ is the dot-product and ${\bf b}$ is a bias term.

\fullpaper{
\vspace{0.4cm}
}
\noindent{\textbf{Label distribution:}}
The encoder function ${\bf {f^2}\left( X \right) = {f}^2_{1:T}}$ encodes
input features $X$ and the function ${\bf {g}^2\left( Y \right) = {g}^2_{1:U}}$
encodes label embeddings. At each time position $\left( t, u \right)$,
the joint score is calculated over all $y \in V$:
\shortpaper{
  \vspace{-0.1cm}
}
\begin{eqnarray}
  S_{t,u} (y | X, y_{0:u}) = {\bf J} \left( {\bf {f}^2_t + {g}^2_u} \right)
  \label{eq:posterior_score}
\end{eqnarray}
\shortpaper{
  \vspace{-0.1cm}
}
where ${\bf J(\cdot)}$ can be any function that maps ${\bf {f}^2_t + {g}^2_u}$ to a $|V|$-dim score vector.
The label posterior distribution is derived by normalizing the score functions
across all labels in $V$:
\shortpaper{
\vspace{-0.2cm}
}
\begin{eqnarray}
  P_{t,u} \left( y_{u + 1}  | X, y_{0:u} \right) = \frac{\exp \left(S_{t,u} (y_{u+1} | X, y_{0:u}) \right)}{\sum_{y \in V} \exp \left(S_{t,u} (y | X, y_{0:u}) \right)}
 \label{eq:hat_label_posterior}
\end{eqnarray}

The alignment path  posterior of the HAT model is computed by chaining the local
edge posteriors of Eq~\ref{eq:hat_model}, $P (\widetilde{Y} | X) = \prod_{k=1}^{T+U-1} P({\widetilde{Y}}_k | X, {\widetilde{Y}}_{1:k-1})$.
\fullpaper{
The posterior of the bold path is:
\begin{eqnarray}
  b_{1,0} b_{2,0} b_{3,0} (1 - b_{4,0}) P_{4,0} ( H |X,  \texttt{<S>}) \cdots b_{5, 3} \nonumber
\end{eqnarray}
}
Like any other time-synchronous model, the total posterior is derived by marginalizing
the alignment path posteriors using Eq~\ref{eq:margin_posterior}.

\shortpaper{
\vspace{-0.3cm}
}
\subsection{Internal Language Model Score}
\label{subsection:ILM}
\shortpaper{
\vspace{-0.2cm}
}
The separation of blank and label posteriors allows the HAT model to produce
a local and sequence-level internal LM score.
The local score at each label position $u$ is defined as:
$S ( y | y_{0:u} ) \triangleq {\bf J} \left({\bf {g}^2_u} \right)$.
In other words, this is exactly the posterior
score of Eq~\ref{eq:posterior_score} but eliminating the effect of 
the encoder activations, ${\bf {f}^2_t}$.
The intuition here is that a language-model quality measure at label $u$ should be only 
a function of
the $u$ previous labels and not the time frame $t$ or the acoustic features.
\ignore{
Note that because of the RNN architectue used for $g$ the activation ${g}^2_u$, depends
on all the $u$ previous labels.
}
Furthermore, this score can be normalized to produce a prior distribution
for the next label $P \left( y | y_{0:u} \right) = \text{softmax} (S ( y | y_{0:u} ))$.
The sequence-level internal LM score is:
\fullpaper{
\begin{eqnarray}
  \log P_{ILM} (Y) = \sum_{u=0}^{U-1} \log P (y_{u+1} | y_{0:u})
  \label{eq:p_ilm}
\end{eqnarray}
}
\shortpaper{
  \vspace{-0.25cm}
\begin{eqnarray}
\vspace{-0.6cm}
  \log P_{ILM} (Y) = \sum_{u=0}^{U-1} \log P (y_{u+1} | y_{0:u})
  \label{eq:p_ilm}
  \vspace{-0.2cm}
\end{eqnarray}
  \vspace{-0.1cm}
}
Note that the most accurate way of calculating the internal language model is
by $P \left( y | y_{0:u} \right) = \sum_{x} P(y | X=x, y_{0:u}) P (X=x | y_{0:u})$.
However, the exact sum is not easy to derive, thus some approximation like
Laplace's methods for integrals \cite{laplace1986memoir} is needed. Meanwhile,
 the above definition can be justified
for the special cases when
${\bf J} \left( {\bf {f}^2_t + {g}^2_u} \right) \approx {\bf J} \left( {\bf {f}^2_t} \right) + {\bf J} \left( {\bf{g}^2_u} \right)$
\shortpaper{
  (proof in arxiv submission.).
}
\fullpaper{
(proof in Appendix~\ref{app:ilm_score}.).
}
\shortpaper{
\vspace{-0.45cm}
}
\subsection{Decoding}
The HAT model inference search for ${\tilde{y}}^{\ast}$ that maximizes:
\shortpaper{
\vspace{-0.2cm}
}
\begin{eqnarray}
         \lambda_1 \log P ( \tilde{y} | x ) - \lambda_2 \log P_{ILM} (B({\tilde{y}}) ) + \log P_{LM} ({\tilde{y}})
  \label{eq:infernece_space}
\end{eqnarray}
where $\lambda_1$ and $\lambda_2$ are scalar weights.
Subtracting the internal LM score, \textit{prior},
from the path posterior leads to a pseudo-likelihood (Bayes' rule), which is justified
for combining with an external LM score.
We use a conventional FST decoder with a decoder graph encoding the phone
context-dependency (if any), pronunciation dictionary and the (external)
n-gram LM. The partial path hypotheses are augmented with the corresponding
state of the model. That is, a hypothesis consists of the time frame $t$,
a state in the decoder graph FST, and a state of the model. Paths with
equivalent history, i.e. an equal label sequence without blanks, are merged
and the corresponding hypotheses are recombined.

%% file: experiment.tex
\shortpaper{
\vspace{-0.4cm}
}
\section{Experiments}
\label{sec:exps}
\shortpaper{
\vspace{-0.15cm}
}
\ignore{
The experiments are designed to evaluate different aspects of the HAT model: overall
training and inference performance, impact of decoding with an external LM
and the effect of context size on the model performance.
}
The training set, $40$ M utterances, development set, $8$ k utterances,
and test set, $25$ hours, are all anonymized,
hand-transcribed, representative of Google traffic queries.
\ignore{
A data set of 40 M utterances
is used for training. A separate development set of
$8$k utterances was used for optimizing hyper-parameters.
The test set is a
separate collection of about $25$ hours from the same domain.
The train, development and test sets are all anonymized,
hand-transcribed, and are representative of Google traffic queries.
}
The training examples are $256$-dim.\ log Mel features extracted from a $64$ ms
window every $30$ ms
\cite{variani2017end}
. The training examples are
noisified with $25$ different noise styles as detailed in
\cite{kim2017efficient}. Each training example is forced-aligned to get the
frame level phoneme alignment.
All models (baselines and HAT) are trained to predict $42$ phonemes
and are decoded with a lexicon and an n-gram language
model that cover a $4$ M words vocabulary. The LM was trained on anonymized audio
transcriptions and web documents.
A maximum entropy LM \cite{biadsy2017effectively} is applied using an
additional lattice re-scoring pass. The decoding hyper-parameters
are swept on separate development set.

Three time-synchronous baselines are presented.
All models use $5$ layers of LSTMs with
$2048$ cells per layer for the encoder. 
For the models with a decoder network (RNNT and HAT), each label is embedded by a
$128$-dim. vector and fed to a decoder network which has $2$ layers of LSTMs with
$256$ cells per layer.
The encoder and decoder activations are projected to a $768$-dim. vector and their
sum is passed to the joint network.
The joint network is a tanh layer followed by a linear
layer of size $42$ and a softmax as in \cite{he2019streaming}.
The encoder and decoder networks are shared for the blank and
the label posterior in the HAT model such that it has exactly the same number of parameters
as the RNNT baseline. The CE and CTC models have $112$M float parameters while RNNT and
HAT models have $115$M parameters.
\ignore{
For all models, the decoding hyper-parameters were swept and best WER result is reported here.
The performance of the all the baseline models along
with the total parameter size is presented in Table~\ref{tab:hat_baseline}.
}
\shortpaper{
\begin{table}[h]
\small
\begin{tabular}{|c|c|c|c|}
\hline
WER            &   \multicolumn{2}{|c|}{\bf 1st-pass}        & {\bf 2nd-pass} \\
\cline{2-3}
(del/ins/sub)& 1st best & 10-best Oracle &  \\
\hline
CE          &  8.9 (1.8/1.6/5.5) & 4.0 (0.8/0.5/2.6)  & 8.3 \\ 
CTC         &  8.9 (1.8/1.6/5.5) & 4.1 (0.8/0.6/2.6)  & 8.4 \\ 
RNNT        &  8.0 (1.1/1.5/5.4) & 2.1 (0.2/0.4/1.6)  &  7.5 \\ 
\hline
HAT        & 6.6 (0.8/1.5/4.3)      & 1.5 (0.1/0.4/1.0) & 6.0 \\ 
\hline
\end{tabular}
\caption{\label{tab:hat_baseline} WER for the baseline and HAT models.}
\end{table}
}
\fullpaper{
\begin{table}[h]
\small
\begin{tabular}{|c|c|c|c|}
\hline
WER            &   \multicolumn{2}{|c|}{\bf 1st-pass}        & {\bf 2nd-pass} \\
\cline{2-3}
(del/ins/sub)& 1st best & 10-best Oracle &  \\
\hline
CE          &  8.9 (1.8/1.6/5.5) & 4.0 (0.8/0.5/2.6)  & 8.3 (1.7/1.6/5.0)\\
CTC         &  8.9 (1.8/1.6/5.5) & 4.1 (0.8/0.6/2.6)  & 8.4 (1.8/1.6/5.0)\\
RNNT        &  8.0 (1.1/1.5/5.4) & 2.1 (0.2/0.4/1.6)  &  7.5 (1.2/1.4/4.9)\\ 
\hline
HAT        & 6.6 (0.8/1.5/4.3)      & 1.5 (0.1/0.4/1.0) & 6.0 (0.8/1.3/3.9) \\
\hline
\end{tabular}
\caption{\label{tab:hat_baseline} WER for the baseline and HAT models.}
\end{table}
}
\shortpaper{
\begin{figure}
  \begin{minipage}[b]{0.23\textwidth}
    \includegraphics[width=\textwidth]{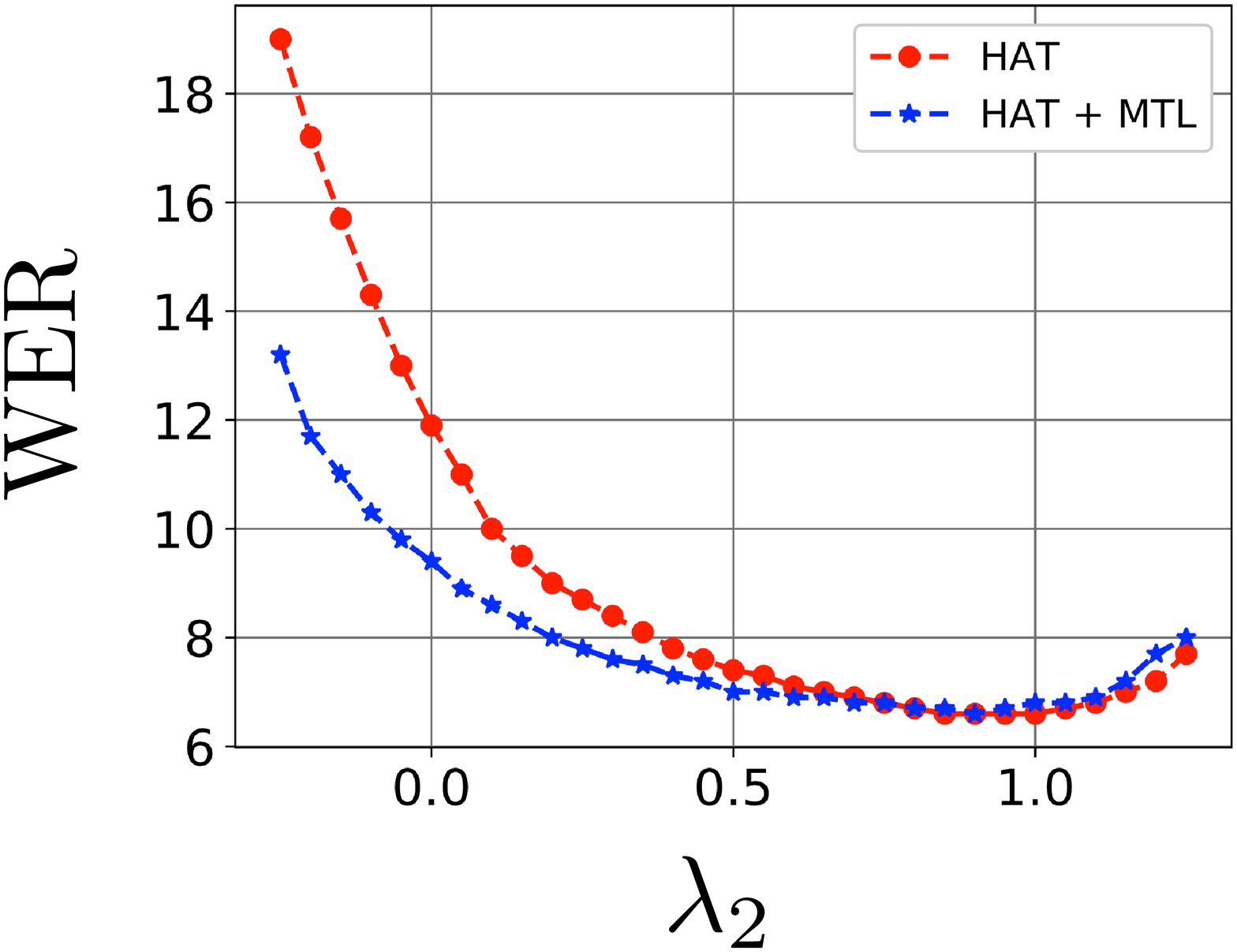}
    \caption{WER vs ILM score weight $\lambda_2$, for $\lambda_1 = 2.5$ for the HAT and HAT + MTL models.}
    \label{fig:lambda2_sweep}
  \end{minipage}
  \hspace{0.1cm}
  \begin{minipage}[b]{0.23\textwidth}
    \includegraphics[width=\textwidth]{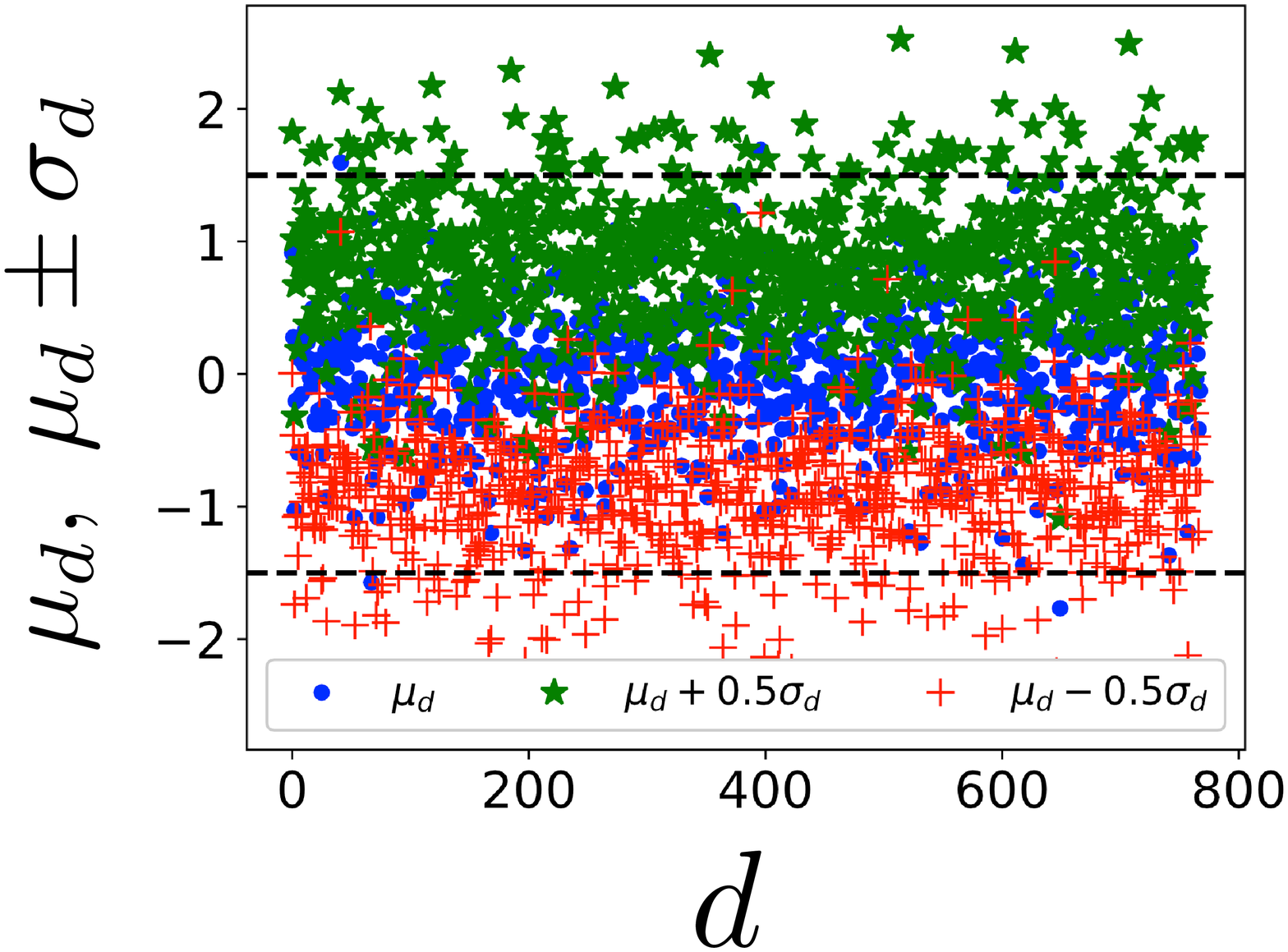}
    \caption{The mean ($\pm$ standard deviation) of ${\bf f_t + g_u}$ falls within linear range of tanh (dotted lines).}
    \label{fig:estimation}
  \end{minipage}
  \vspace{-0.4cm} 
\end{figure}
\begin{figure}
\vspace{-0.3cm}
  \hspace{-0.2cm}
  \includegraphics[width=0.5\textwidth]{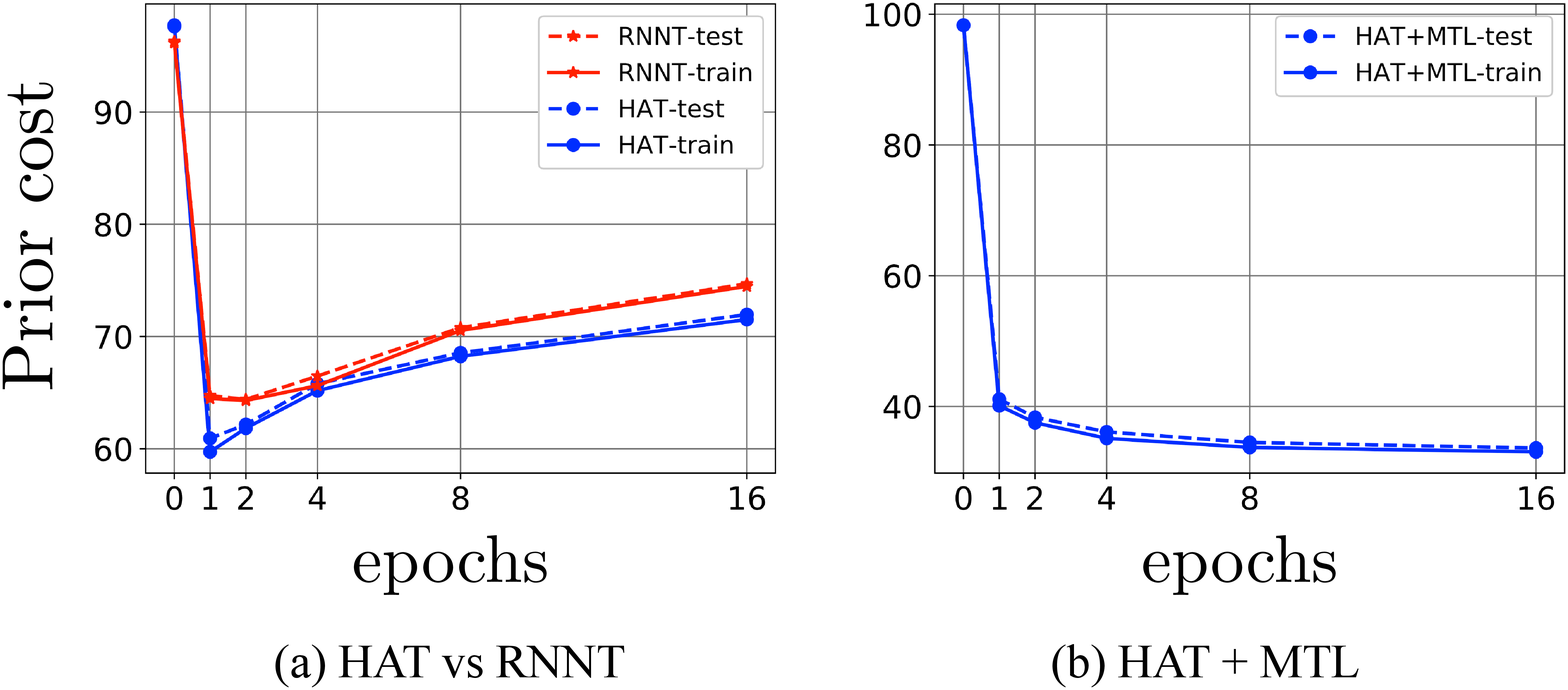}
  \caption{Prior cost vs epochs for HAT, RNNT, and HAT + MTL.}
  \label{fig:prior_cost}
\vspace{-0.6cm}
\end{figure}
\vspace{-0.3cm}
}
\fullpaper{
\begin{figure}
  \begin{minipage}[b]{0.23\textwidth}
    \includegraphics[width=\textwidth]{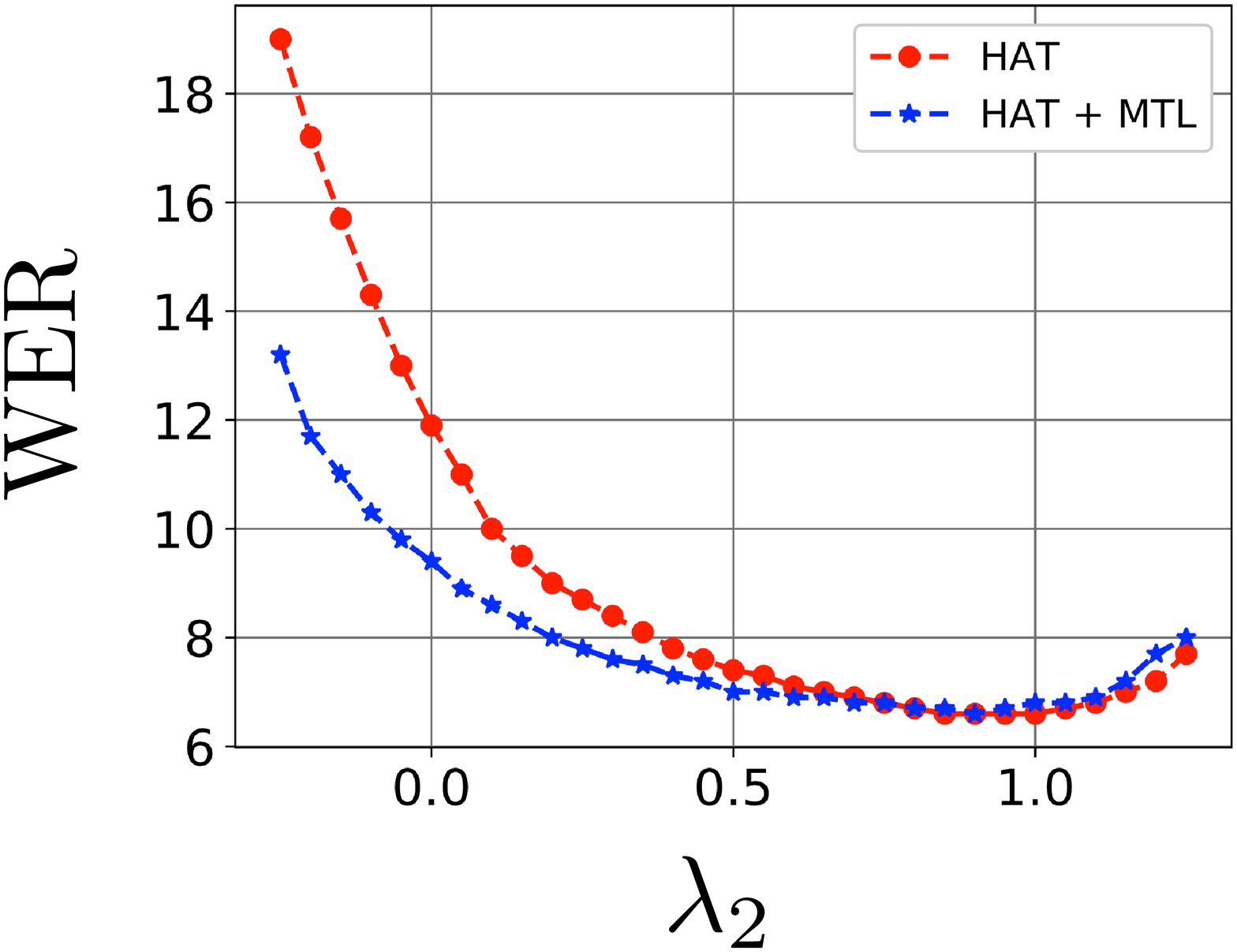}
    \caption{WER [$\%$] vs internal language model score weight $\lambda_2$, for $\lambda_1 = 2.5$ for the HAT ($\color{red} \bullet$) and HAT + MTL ($\color{blue} \ast$) models.}
    \label{fig:lambda2_sweep}
  \end{minipage}
  \hspace{0.1cm}
  \begin{minipage}[b]{0.23\textwidth}
    \includegraphics[width=\textwidth]{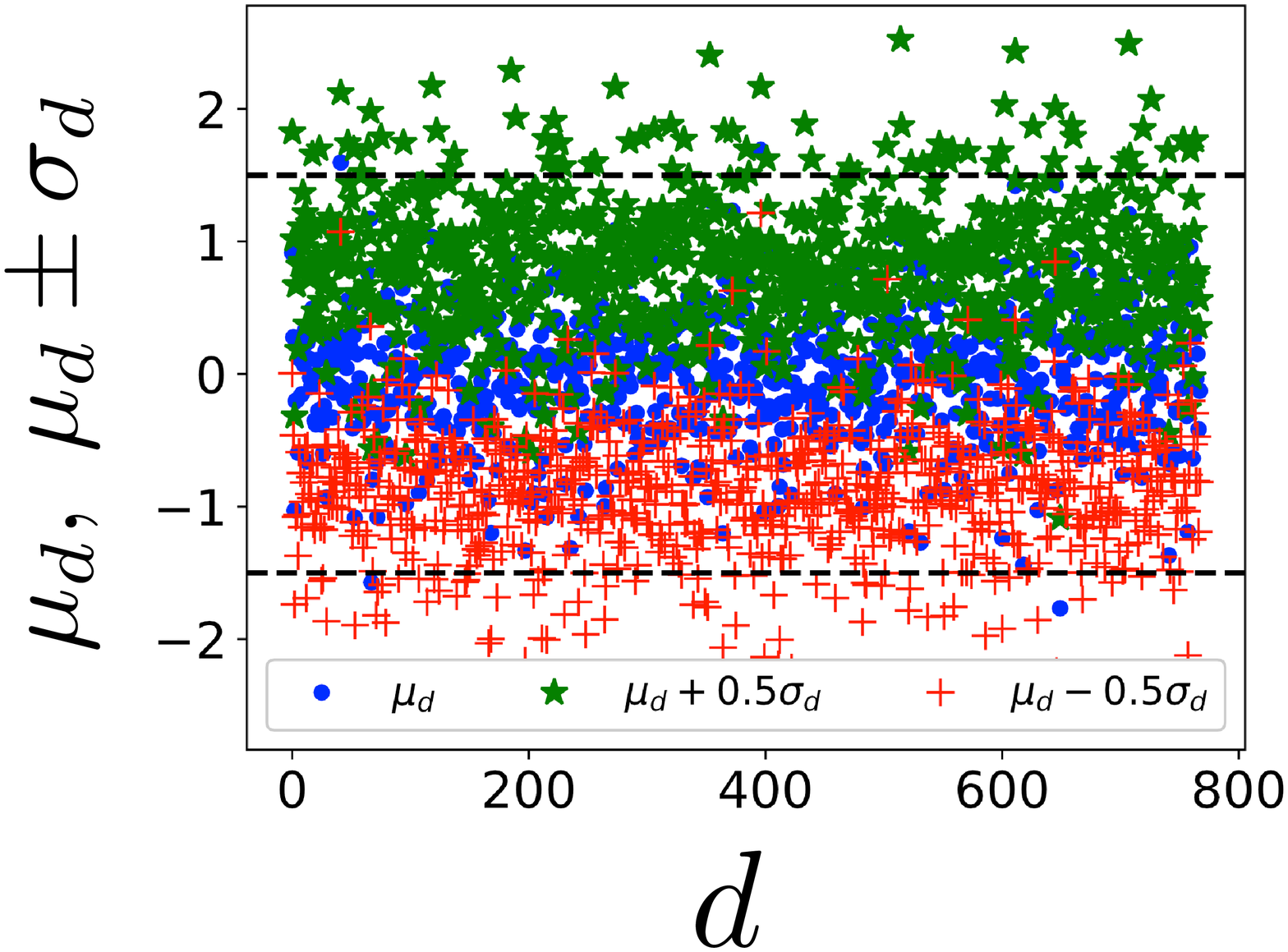}
    \caption{The mean ($\pm$ standard deviation) of ${\bf f_t + g_u}$ falls within linear range of tanh (horizontal dotted lines).}
    \label{fig:estimation}
  \end{minipage}
\end{figure}
\begin{figure}
  \hspace{-0.2cm}
  \includegraphics[width=0.5\textwidth]{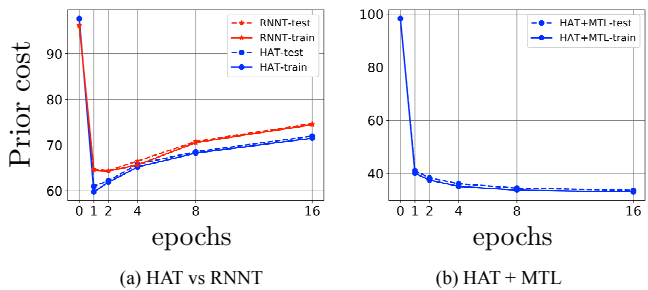}
  \caption{Prior cost vs epochs for HAT, RNNT, and HAT + MTL.}
  \label{fig:prior_cost}
\end{figure}
}

\ignore{
\begin{figure*}[h]
  \includegraphics[height=4.0cm,width=18.0cm]{figures/prior_posterior_heatmap.eps}
  \caption{The Prior distribution $P (y_{u+1} | y_{0:u})$ for $u =1:15$ (left) and the Posterior distribution $P(y_t | X)$ for $t=1:104$(right).}
  \label{fig:heatmaps}
\end{figure*}
}
\vspace{-0.2cm}
\noindent \textbf{HAT model performance:}
For the inference algorithm of Eq~\ref{eq:infernece_space}, we empirically observed that
$\lambda_1 \in (2.0, 3.0) $ and $\lambda_2 \approx 1.0$ lead to the best performance.
Figure~\ref{fig:lambda2_sweep} plots the WER for different values of $\lambda_2$,
for baseline HAT model and HAT model trained with multi-task learning (MTL).
The best WER at $\lambda_2=0.95$ on the convex curve
shows that significantly down-weighting the internal LM and relying more on the external LM yields the best decoding quality.
The HAT model outperforms the baseline RNNT model by $1.4\%$ absolute WER ($17.5 \%$
relative gain), Table~\ref{tab:hat_baseline}. Furthermore, using a 2nd-pass LM
gives an extra $10 \%$ gain which
is expected for the low oracle WER of $1.5 \%$.
\fullpaper{
Comparing the different types of errors,
it seems that the HAT model error pattern is different from all baselines. In all baselines,
the deletion and insertion error rates are in the same range, whereas the HAT model makes two times more insertion than deletion errors.\footnote{We are examining this
  behavior and hoping to have some explanation for it in the final version of
  the paper.}
}

\noindent \textbf{Internal language model:}
The internal language model score proposed in subsection \ref{subsection:ILM}
can be analytically justified when
${\bf J} \left( {\bf {f}^2_t + {g}^2_u} \right) \approx {\bf J} \left( {\bf {f}^2_t} \right) + {\bf J} \left( {\bf{g}^2_u} \right)$.
The joint function ${\bf J}$
is usually a tanh function followed by a linear layer \cite{he2019streaming}.
The approximation holds
iff ${\bf {f}^2_t + \bf{g}^2_u}$ falls in the linear range of the tanh function.
Figure~\ref{fig:estimation} presents the statistics of this $768$-dim. vector,
$\mu_d$ (${\color{blue} \bullet}$), $\mu_d + \sigma_d$ (${\color{green} \star}$),
and $\mu_d - \sigma_d$ (${\color{red} +}$) for $1 \leq d \leq 768$ which are
accumulated on the test set. The linear range of tanh function is specified
by the two dotted horizontal lines in Figure~\ref{fig:estimation}. The means and large portion
of one standard deviation around the means fall within the linear range of the tanh,
which suggests that the decomposition of
the joint function into a sum of two joint functions is plausible.

Figure~\ref{fig:prior_cost}(a) shows the
\textit{prior cost}, $\frac{-1}{|\mathcal{D}|} \sum_{y \in \mathcal{D}} \log P_{ILM} (y)$,
during the first $16$ epochs of training with $\mathcal{D}$ being either the train
or test set. The curves are plotted for both the HAT and RNNT models. Note that in case
of RNNT, since blank and label posteriors are mixed together, a softmax
is applied to normalize the non-blank scores to get the label posteriors, thus it just serves
as an approximation.
At the early steps of training, the prior cost is
going down, which suggests that the model is learning an internal language model. However,
after the first few epochs, the prior cost starts increasing, which suggests that
the decoder network deviates from being a language model.
Note that both train and test curves behave like this.
One explanation of this observation can be:
in maximizing the posterior, the model
prefers not to choose parameters that also maximize the prior.
We evaluated the performance of the HAT model when it further constrained with a
cross-entropy multi-task learning (MTL) criterion that minimizes the prior cost.
Applying this criterion results in decreasing the prior
cost, Figure~\ref{fig:prior_cost}(b). However, this did
not impact the WER. After sweeping $\lambda_2$, the WER for the HAT model with
MTL loss is as good as the baseline HAT model, blue curve in Figure~\ref{fig:lambda2_sweep}.
Of course this might happen because the internal language model is still weaker
than the external language model used for inference.

This observation might also be explained by Bayes' rule:
$
\log P (X | Y ) \propto \log P (Y | X) - \log P (Y)
$.
The observation that the posterior cost $- \log P (Y | X)$ goes down
while the prior cost $- \log P (Y)$ goes up
suggests
that the model is implicitly maximizing the log likelihood term in the left-side
of the above equation. This might be why subtracting the internal language
model log-probability from the log-posterior and replacing it with a strong external
language model during inference led to superior performance.

\noindent \textbf{Limited vs Infinite Context:}
Assuming that the decoder network is not taking advantage of the
full label history, a natural question is how much of the label
history actually is needed to achieve peak performance. The HAT model
performance for different context sizes is shown in
Table~\ref{tab:hat_context}.
A context size of $0$ means feeding no label history to the decoder network,
which is similar to making a conditional independence assumption
for calculating the
posterior at each time-label position. The performance of this model is on a par with
the performance of the CE and CTC baseline models, c.f.\ Table~\ref{tab:hat_baseline}.
The HAT model with context size of $1$ shows $12 \%$ relative WER degradation compared
to the infinite history HAT model. However, the HAT models with
contexts $2$ and $4$ are
matching the performance of the baseline HAT model.

While the posterior
cost of the HAT model with context size of $2$ is about $10 \%$ worse than the baseline HAT
model loss, the WER of both models is the same. One reason for this behavior is that the
external LM has compensated for any shortcoming of the model. Another explanation
can be the problem of {\em exposure bias} \cite{ranzato2015sequence}, which refers to the mismatch between
training and inference for encoder-decoder models. The error propagation in an infinite
context model can be much severe than a finite context model with context $c$,
simply because the finite context model has a chance to reset its decision every
$c$ labels.
\shortpaper{
\vspace{-0.1cm}
}
\begin{table}[h]
\begin{center}
\small
\begin{tabular}{|c|c|c|c|c|c|c|c|c|}
\hline
Context         & 0 & 1 & 2 & 4 & $\infty$\\
\hline
1st-pass WER        & 8.5 &  7.4  &  6.6  & 6.6 & 6.6  \\
posterior cost            & 34.6 &  5.6  & 5.2   & 4.7 & 4.6  \\
\hline
\end{tabular}
\caption{\label{tab:hat_context} Effect of limited context history.}
\end{center}
\shortpaper{
\vspace{-0.3cm}
}
\end{table}
\vspace{-0.5cm}

\ignore{
Another interesting aspect of the finite context models is the possibilities they
provide for training and inference.
In the above experiments, given that}
Since the finite
context of $2$ is sufficient to perform as well as an infinite context, one can simply
replace all the expensive RNN kernels in the decoder network with a $|V|^2$ embedding
vector corresponding to all the possible permutations of a finite context of size
$2$. In other words, trading computation with memory, which can significantly reduce
total training and inference cost.

%% file: conclusion.tex
\shortpaper{
\vspace{-0.3cm}
}
\section{Conclusion}
\label{sec:conclusion}
\shortpaper{
\vspace{-0.2cm}
}
The HAT model is a step toward better acoustic modeling while
preserving the
modularity of a speech recognition system.
One of the key advantages of the HAT model is the introduction
of an internal language model quantity that can be measured
to better understand encoder-decoder models and
decide if equipping them with an external language
model is beneficial. According to our analysis, the decoder
network does not behave as a language model but
more like a finite context model. We presented a few
explanations for this observation. Further in-depth analysis
is needed to confirm the exact source of this behavior and how
to construct models that are really end-to-end, meaning the 
prior and posterior models behave as expected for a language model and
acoustic model.

%% file: appendix.tex
\section{Internal Language Model Score}
\label{app:ilm_score}

\begin{lemma}
  The softmax function is invertible up to an additive constant. In other words,
  for any two real valued vectors ${\bf v=v_{1:d}}, {\bf w = w_{1:d}}$, and real
  constant $c$,
  \[
   v_i = w_i + c  \Longleftrightarrow \frac{\exp (v_i)} {\sum_j \exp (v_j)} = \frac{\exp(w_i)}{\sum_j \exp(w_j)}
   \]
  for $1 \leq i \leq d$.
  \label{lemma1}
\end{lemma}

\begin{proof}
  If $\forall i: \,\,v_i = w_i + c\,$, then
  \begin{eqnarray}
    \frac{\exp (v_i)} {\sum_j \exp (v_j)} &=& \frac{\exp (w_i + c)} {\sum_j \exp (w_j + c)} \nonumber \\
    &=& \frac{\exp (w_i) \exp(c)}{\sum_j \exp (w_j) \exp(c)} \nonumber \\
    &=& \frac{\exp (w_i)}{\sum_j \exp (w_j)}
  \end{eqnarray}
  which proves the if condition. For the other side of condition, the proof goes as follow:
  \[
  \frac{\exp (v_i)} {\sum_j \exp (v_j)} = \frac{\exp(w_i)}{\sum_j \exp(w_j)}
  \]
  taking logarithm from both sides:
  \[
    v_i - \log \sum_j \exp(v_j) = w_i - \log \sum_j \exp(w_j)
  \]
  which implies:
  \[
  v_i - w_i = \log \sum_j \exp(v_j) - \log \sum_j \exp(w_j)
  \]
  this completes the proof since the right-hand-side (RHS) of above equation is constant for any $i$.
\end{proof}

\begin{corollary}
  For any two random variables $X$ and $Y$ and some real valued function
  $S (y, x)$, if the posterior distribution of $Y$ given $X$ be
  \[
  P (Y = y | X = x) = \frac{\exp (S (y, x))} {\sum_{y'} \exp S(y', x)}
  \]
  then $S (y, x) = \log P (Y=y, x) + c \,$, for some constant value $c \in \mathbbm{R}$.
  \label{corollary1}
\end{corollary}

\begin{proof}
  The proof is straight-forward following Lemma~\ref{lemma1}.
\end{proof}
\begin{proposition}
For the label distribution of Eq~\ref{eq:hat_label_posterior} with the score
function of Eq~\ref{eq:posterior_score}, if
${\bf J} \left( {\bf {f}^2_t + {g}^2_u} \right) \approx {\bf J} \left( {\bf {f}^2_t} \right) + {\bf J} \left( {\bf{g}^2_u} \right)$
for any $t$ and $u$,
then:
\begin{eqnarray}
  P (y | y_{0:u}) \propto \exp \left( {\bf J} \left({\bf {g}^2_u} \right) \right)
\end{eqnarray}
\end{proposition}

\begin{proof}
According to Eq~\ref{eq:hat_label_posterior}
\begin{eqnarray}
  P \left( y  | X=x, y_{0:u} \right) = \frac{\exp \left(S_{t,u} (y | X=x, y_{0:u}) \right)}{\sum_{y' \in V} \exp \left(S_{t,u} (y | X=x, y_{0:u}) \right)}
\end{eqnarray}
which implies:
\begin{eqnarray}
  S_{t,u} (y | X=x, y_{0:u}) &=& \log P\left(y, X=x, y_{0:u}\right) + c_1 \nonumber \\
  &=& \log P\left(X=x | y, y_{0:u}\right) + \log P\left(y, y_{0:u}\right) \nonumber \\
\end{eqnarray}
for some real valued constant $c_1$. The first equality holds from Corollary~\ref{corollary1},
and the second equality holds by Bayes' rule.

\noindent Applying exponential function and marginalizing
over $X$ results in:
\begin{eqnarray}
  \sum_{X} \exp ( S_{t,u} (y | X=x, y_{0:u}) ) &\propto& \sum_{x} P\left(X=x | y, y_{0:u}\right) P\left(y, y_{0:u}\right) \nonumber \\
  &=& P\left(y, y_{0:u}\right)
  \label{eq:first_eq}
\end{eqnarray}
note that $\exp(c_1)$ is dropped which make the left-hand-side (LHS) be proportional to
the RHS of first equation. The second equation holds since $P\left(X=x | y, y_{0:u}\right)$
is a distribution over $X$.

\noindent Finally note that
\begin{eqnarray}
  S_{t,u} (y | X, y_{0:u}) &=& {\bf J} \left( {\bf {f}^2_t + {g}^2_u} \right) \nonumber \\
  &=& {\bf J} \left( {\bf {f}^2_t} \right) + {\bf J} \left( {\bf{g}^2_u} \right)
\end{eqnarray}
where the first equality comes from the definition and the second one is the assumption
made in the statement of the proposition. Applying exponential function and marginalizing over
$X$:
\begin{eqnarray}
  \sum_{X} \exp ( S_{t,u} (y | X, y_{0:u}) ) &=& \sum_{X} \exp \left({\bf J} \left( {\bf {f}^2_t} \right) + {\bf J} \left( {\bf{g}^2_u} \right) \right) \nonumber \\
  &=& \exp \left({\bf J} \left( {\bf{g}^2_u} \right) \right) {\sum_{X} \exp \left({\bf J} \left( {\bf {f}^2_t}  \right) \right)} \nonumber \\
  &\propto& \exp \left({\bf J} \left( {\bf{g}^2_u} \right) \right)
  \label{eq:second_eq}
\end{eqnarray}
since the second term will be a constant and will not be function of $y$.
Comparing Eq~\ref{eq:first_eq} and Eq~\ref{eq:second_eq}:
\[
P\left(y, y_{0:u}\right) \propto \exp \left({\bf J} \left( {\bf{g}^2_u} \right) \right)
\]
\noindent Finally note that $P \left(y | y_{0:u} \right) = P\left(y, y_{0:u}\right) / P\left(y_{0:u}\right) $, thus:
\[
P \left(y | y_{0:u} \right) \propto \exp \left({\bf J} \left( {\bf{g}^2_u} \right) \right)
\]
which completes the proof.
\end{proof}